\documentclass{article}
\usepackage{multirow}
\usepackage{graphicx}

\usepackage{arxiv}

\usepackage{hyperref}       

\usepackage{balance}
\usepackage{caption}
\usepackage{color}
\usepackage[a-1b]{pdfx}
\usepackage{subcaption}
\usepackage{amsthm}
\usepackage[ruled,vlined]{algorithm2e}
\usepackage{amsmath}
\usepackage{stfloats}

\newtheorem{theorem}{Theorem}
\newtheorem{lemma}{Lemma}
\theoremstyle{definition}
\newtheorem{definition}{Definition}
\theoremstyle{remark}

\theoremstyle{remark}

\theoremstyle{remark}

\newcommand{\bpt}{B$^{+}$-tree}
\newcommand{\clrs}{CLRS-B$^{+}$-tree}
\newcommand{\sys}{FFBtree}

\newcommand{\Nil}{\mathsf{nil}}
\newcommand{\rootnode}{\mathsf{root}}
\newcommand{\FindNext}{\mathsf{FindNext}}
\newcommand{\FindSibling}{\mathsf{FindSibling}}
\newcommand{\Split}{\mathsf{ProactiveSplit}}
\newcommand{\Merge}{\mathsf{Merge}}
\newcommand{\Insert}{\mathsf{Insert}}
\newcommand{\Delete}{\mathsf{Delete}}
\newcommand{\Borrow}{\mathsf{Borrow}}

\newcommand{\F}{\mathsf{F}}
\newcommand{\U}{\mathsf{Util}}
\newcommand{\Sfun}{\mathsf{s}}
\newcommand{\occ}{\operatorname{Occ}}
\newcommand{\capf}{\operatorname{Cap}}
\newcommand{\children}{\operatorname{children}}
\newcommand{\CritSet}{\mathcal{C}}
\newcommand{\UnsafeSet}{\mathcal{U}}
\newcommand{\NeedSet}{\mathcal{S}}

\SetKwComment{Comment}{/* }{ */}
\title{Towards a \bpt{} with Fluctuation-Free Performance}

\author{Lu Xing, Walid G. Aref\\
	Department of Computer Science\\
	Purdue University\\
	West Lafayette, IN 47906 \\
	\texttt{\{xingl, aref\}@purdue.edu} \\
}

\date{}

\begin{document}
\maketitle

\begin{abstract}
Performance predictability is critical for modern DBMSs. Index maintenance can trigger rare but severe I/O spikes. In a  B/\bpt{} of Height $H$, due to {\em node split propagation},  the cost of a single insert can vary from $(H+1)$ I/Os when node split does not happen to $(3H+1)$ I/Os when split propagation reaches the root, nearly a $3\times$ degradation. We formalize performance fluctuation as the gap between the best- and worst-case insert behaviors. We show examples of why performance fluctuation is worth investigating. To 
mitigate
the B/\bpt{} split propagation, we introduce the \sys{}, a \bpt{} with a novel insert and delete algorithms with bounded-fluctuation performance. We prove that the insert operation is guaranteed to  encounter at most one node that needs to be split, and hence ensuring that no split propagation can take place. 
The new insert and delete algorithms guarantee bounded-fluctuation performance for the \sys{}. Experiments with simulated indexes show that the \sys{} caps I/O fluctuation compared to baseline indexes under various workloads.
\end{abstract}

\section{Introduction}\label{intro}



Performance predictability is crucial for many applications. In real-time applications, e.g., autonomous cars and robots, when an unseen scenario happens, the performance of the systems may become unpredictable in terms of both time and decisions. 
The risk level is elevated, e.g., when 
stacks of multiple components have fluctuating performance. 
This problem can be partially solved by synthesizing many scenarios to simulate rare edge cases~\cite{waymo}. Since it is hard to enumerate {\it all} rare edge cases, the
goal is to design a system that is proved to perform in a  deterministic way under all possible cases.

\begin{sloppypar}
Modern database management systems (DBMSs) must deliver not only high transaction throughput but also predictable performance. As Dean and Barroso note~\cite{dean2013tailatscale}, ``Large online services need to create a predictably responsive whole out of less-predictable parts." Consider an e-commerce platform: even infrequent delays in processing user transactions can noticeably degrade user experience. Similarly, when a service provider commits to a service-level objective (SLO)~\cite{jones2016service}, a small number of unexpected spikes may trigger violations, harming both  provider and  customer. Variance is  amplified at scale~\cite{dean2013tailatscale}; therefore, DBMS performance should be evaluated not only using averages but also targeting low-variance and stable latencies. 
\end{sloppypar}

\begin{sloppypar} 
Latency spikes often arise from rare yet expensive internal DBMS events. Index maintenance provides a representative example~\cite{dean2013tailatscale,chen2020utree,balmau2019silk,glombiewski2019waves,xing2021vldb}. The B/\bpt{}~\cite{bayer1970organization}, one of the most widely used indexing structures, maintains balance by splitting nodes upon overflow and inserting separator keys into parent nodes. In the worst case, a split can trigger {\it split propagation}, cascading from the leaf node to be split all the way up to the root of the tree. This cascade is costly because it allocates new pages and forces additional page writes, potentially incurring multiple additional disk I/Os, and this happens within only one insert. It can also perturb the buffer pool (e.g., evictions and write-back pressure), indirectly slowing down concurrent reads and amplifying user-visible latency. 
\end{sloppypar}
\begin{sloppypar}
Other data indexes can also experience fluctuation in performance. For example, 
Log-Structured Merge-Trees (LSM-trees)~\cite{o1996log} exhibit a similar phenomenon. When a level reaches capacity, the system performs compaction that reads multiple files from disk, merge-sorts them (often with significant CPU and memory overhead), and writes new files back to disk. Because compaction consumes substantial I/O bandwidth and compute resources, it can interfere with foreground operations and manifest as latency spikes or throughput drops~\cite{balmau2019silk,balmau2020silk}. More broadly, performance fluctuations are not limited to data indexes. Other subsystems, e.g., buffer management, transaction management, background daemons, global resource sharing, and maintenance activities can introduce additional, hard-to-predict stalls~\cite{dean2013tailatscale}.
\end{sloppypar}

A substantial body of prior work improves amortized DBMS performance, often by batching or deferring maintenance to reduce average I/O and increase throughput~\cite{o1996log,arge1995buffer,bender2015introduction}. However, DBMSs must also control performance variance, including fluctuations in latency and disk I/O. This variability is often hidden by averaged metrics, and techniques that improve mean performance do not necessarily reduce variance; in some cases, tail behavior is unchanged.

\begin{sloppypar}
We pursue index designs with predictable, ideally deterministic and bounded-fluctuation performance. Achieving this goal raises several challenges. First, we must formalize what {\it fluctuation} means and how it should be quantified in terms of I/O or latency. Second, predictability should not come at the expense of efficiency. Forcing every insert to pay a worst-case cost would remove variance but would be impractical. Instead, we seek algorithms that minimize fluctuation while preserving average performance comparable to conventional designs. Third, a formal proof is needed for guarantees of bounded fluctuation. Forth, besides a formal proof, the algorithms should be tested under adversary workloads.
\end{sloppypar}

In this paper, we study performance fluctuation in index structures, focusing on \bpt{} indexes. We formalize fluctuation and show that the conventional \bpt{} insert and delete algorithms are not fluctuation-bounded. Motivated by the observation that certain nodes eventually become unavoidable split points, we introduce the concepts of {\em safe} and {\em critical} nodes to characterize when a split can be triggered. Building on these definitions, we propose a new insert algorithm that initiates splits earlier in a controlled manner. The new \sys{} eliminates split propagation and guarantees that each insert performs at most one node split while descending from root to leaf and is not forced to perform additional splits. 
Hence, this bounds the number of node splits during an insert operation to at most one node split per insert, and hence bounds the fluctuation in performance to a constant. We also propose a new delete algorithm that bounds the performance by preventing merge propagation.

The contributions of this paper are as follows:
\vspace{-\topsep} 
\begin{itemize}
    \item We demonstrate that practical tree indexes exhibit measurable performance fluctuation, and we provide a formal definition of fluctuation for tree-based indexes.
    \item We introduce the \sys{} that is equipped with a new insert algorithm with a provable guarantee that each insert triggers \emph{at most one} node split.
    \item We introduce a novel delete algorithm that bounds the number of needed I/Os.
    \item We verify our algorithms using simulation, and compare with baselines under several different datasets.
\end{itemize}

\section{Related Work}

B-trees have been introduced in 1970 by Bayer and McCreight~\cite{bayer1970organization}. Since then, they have inspired extensive research and numerous variants. Comer~\cite{comer1979ubiquitous} provides an early survey of B-tree variants. Graefe~\cite{graefe2011modern,graefe2024more} offers a more recent overview of modern B-tree techniques.

One way to improve the B/\bpt{} fluctuation is to optimize its write performance. There have been many research works on that area. The Buffer tree~\cite{arge1995buffer} is a B-tree that has a buffer associated with each tree node. The buffer stores data and flushes once full, which amortizes  insert cost. A similar structure, the B$\epsilon$-tree~\cite{bender2015introduction} also stores data in a buffer allocated by the non-leaf node. The parameter $\epsilon$ can be tuned to distribute the space between the key pivots and the buffer. SplinterDB~\cite{conway2020splinterdb} adopts the idea of the B$\epsilon$-tree and extends it with an LSM-tree. The FD-tree~\cite{li2010vldb} is a write-optimized B-tree designed for solid-state drives (SSDs). It consists of a head tree on top and several levels of sorted runs beneath. Optimizing B/\bpt{} write performance alone cannot stabilize the variance between best- and worst-case performance.

Reducing tail latency  (e.g., the 99th-percentile response time) can help making the system response more consistent, i.e., less fluctuating.
Previous research has explored that tail latency often arises from rare but expensive events, e.g., contention, background maintenance, or device-level stalls. \cite{dean2013tailatscale} states that a slow operation harms the service responsiveness that is amplified by scale. 
The uTree~\cite{chen2020utree} targets persistent memory and reduces tail latency by adding a shadow-list layer at the leaf level, combining the update-friendliness of lists with the search efficiency of trees. 
For LSM-tree-based key-value stores that show periodic latency peaks during file compaction, SILK~\cite{balmau2019silk} and SILK+~\cite{balmau2020silk} use scheduling approach to constrain tail behavior.
In transaction processing, Plor~\cite{chen2022plor} combines pessimistic locking with optimistic reading in a hybrid protocol.
Polaris~\cite{ye2023polaris} adds priority enforcement to optimistic concurrency control with a small amount of pessimism.
Tiny-Tail Flash~\cite{yan2017tiny} targets garbage-collection-induced tail latencies inside SSDs, while RAIL~\cite{litz2022rail} aims for predictable low tail read latency on NVMe flash via the internal parallelism of the devices. 
Perséphone~\cite{demoulin2021idling} reduces tail latency for heavy-tailed microsecond-scale workloads using a kernel-bypass scheduler that reserves resources for short requests rather than being strictly work-conserving. The above systems target specific hardware or transaction processing, but the  B/\bpt{} is not theoretically analyzed.

Enforcing time constraints of operations helps reduce the fluctuation as all operations have to be finished within a certain timer limit. That has been explored in real-time system design. 
\cite{liu1973scheduling} establishes a mathematical foundation for scheduling tasks in a hard real-time system~\cite{manacher1967production}, where tasks need to meet strict timing constraints. 
Work on selecting periods/deadlines shows how tuning system parameters can preserve feasibility guarantees~\cite{chantem2007period}. 
\cite{takado2024real} proposes a real-time operating system (RTOS) for cyber-physical systems with communication time fluctuations~\cite{takado2024real}. 
LazyTick~\cite{heider2025lazytick} reduces the overhead and variability in how RTOSes release tasks by partitioning the task set and spreading tasks across multiple timers. 
\cite{wilhelm2008worst} surveys the worst-case execution time (WCET) problem and reviews key analysis techniques and compares existing WCET analysis tools. Real-time database systems (RTDBS) provide database operations and real-time constraints and are surveyed in~\cite{kao1994overview}; \cite{pang1994query} focuses on the query processing component of RTDBS.

The fluctuation-bounded performance focuses on the same problem as the de-amortization of algorithms~\cite{rao1998amortization}. De-amortization requires an algorithm that costs $f(n)$ steps for $n$ operations to take no more than $O(f(n)/n)$ per operation. This is equivalent to bounding the worst-case performance among $n$ operations. We will discuss this further in more detail in Section~\ref{subsec:algo-cost}.

Finally, note that we use the word {\em deterministic} as well as {\em bounded-fluctuation} interchangeably. Deterministic data-index performance should not be conflated with deterministic database systems, as surveyed in~\cite{abadi2018overview}. A deterministic DBMS focuses on transaction scheduling, i.e., it guarantees that, given the same input transactions, the system will always reach the same final state. This line of work was originally motivated by database replication~\cite{abadi2018overview,jimenez2000deterministic,kemme2000don,schneider1990implementing,thomson2010case}, and later is shown to offer benefits in scalability, modularity, and concurrency~\cite{abadi2018overview}.


\section{Problem Setting}
Performance and resource usage often fluctuate in practice. In this section, we define fluctuation and bounded fluctuation behavior. We explain why existing \bpt{} algorithms are insufficient to address the fluctuation problem.

\subsection{Preliminaries}\label{subsec:problem-def}

\begin{sloppypar}
We define {\it fluctuation} as the gap between maximum (worst-case) and minimum (best-case) performance, where performance can refer to latency, number of I/Os, or throughput (Definition~\ref{def:fluct}).
\end{sloppypar}

\begin{definition}\label{def:fluct}
    {\bf (Fluctuation).} Let $B$ and $W$ be the best- and worst-case performances of an operation $O$, respectively. Fluctuation $\Phi$ of Operation $O$ is defined as $\Phi(O)=|W-B|$, representing the difference between the worst- and best-case performances.
\end{definition}

The larger the fluctuation, the more unstable the system is. For example, if the fluctuation of latency is large in an online transactional system (OLTP), it indicates that the system may fail to meet the service level agreement (SLA) guarantees on sustained performance. 

\begin{definition}\label{def:bounded-fluct}
    {\bf (Bounded Fluctuation).} Let $n$ be the dataset size. Average performance $P$ of Operation $O$ can be modeled as $P_{\text{op}=O}=f(n)$. We say that fluctuation of $P$ is bounded iff $\Phi(O)=g(n)$ where $g(n)\in o(f(n))$.
\end{definition}

\begin{sloppypar}
Bounded fluctuation captures cases where variability grows asymptotically slower than the baseline cost. When the size of the dataset scales, the variability does not scale with the dataset size. This is advantageous for the big data that deals with enormous size of data. For example, the \bpt{}'s search performance is $H$ I/Os for both the worst- and best-cases. Therefore, $\Phi(\text{read})=0$ and the \bpt{}'s search is bounded. In another example of merge sort, if the performance is measured by the number of elements comparison. When merging two subarrays, the best-case is when all elements in one subarray are smaller or larger than the ones in another subarray; while the worst-case requires all elements interleaved. $B=n/2\log n$ and $W=n\log n - n + 1$, thus $\Phi=|W-B|=n/2\log n-n+1$ is not bounded.
\end{sloppypar}

\begin{sloppypar}
{\bf Caveat.} Our definition of {\it fluctuation} is based on best- and worst-case performance. One might attempt to achieve bounded fluctuation by maintaining a constant, yet suboptimal, lower bound of performance, thereby eliminating variability. However, this approach renders the system unusable due to persistent inefficiency. Hence, we define fluctuation in terms of both best- and worst-case performance to ensure meaningful evaluation.
\end{sloppypar}

\subsection{\bpt{} Discussion}

The \bpt{}~\cite{bayer1970organization,comer1979ubiquitous} is one of the most widely used indexes in database management systems. The \bpt{} is a self-balancing tree structure where each tree node is of one page size that stores more than two children nodes. It supports logarithmic insert and lookup operations asymptotically.

\noindent
{\bf Insertion}: We investigate the insertion process of the \bpt{} closely. We make the following assumptions: all tree nodes are stored on disk; for each insertion, the corresponding nodes are loaded into memory, and the nodes already brought into memory are retained until the insertion completes, after which they are flushed back to disk. Suppose the insertion of $d_1$ goes into a leaf node of a height $H$ tree. The total incurred I/O is $(H + 1)$ including $H$ read I/Os of the nodes along the descending path and one write I/O of the leaf node that $d_1$ targets. This is considered to be the best-case or minimum insertion cost for a data item into the \bpt{}.

Still the same tree, another insertion of $d_2$ targets a leaf node that is full, whose parent node is also full, the cost is increased because both the leaf node and the parent node have to be split to accommodate $d_2$. Inserting $d_2$ triggers a {\it split propagation} of two nodes. An even more catastrophic split propagation may involve up to the root node. The total cost of I/O would be $(3H+1)$ including $H$ read I/Os of the nodes along the descending path, and two write I/Os per level, plus one write I/O of the new root node. This is considered the worst-case or maximum insertion cost for a data item into the \bpt{}.

\begin{sloppypar}
The fluctuation of the \bpt{} insertion is $\Phi(\text{insert})=(3H+1)-(H+1)=2H$, which implies $\Phi \in O(H)$ and $\Phi\notin o(H)$. Since $\Phi$ grows at the same rate as $H$, the insertion fluctuation is not bounded.
\end{sloppypar}

\begin{sloppypar}
\noindent
{\bf Deletion}: When deleting a data item, say $d$ from the \bpt{}, if the target leaf node, say $N$, is not underflowing, i.e., node utilization is more than half-full, then $d$ is deleted from $N$ and no further action is needed. When the target leaf node underflows, the first strategy is to {\it borrow} from a non-underflowing sibling node. Borrowing keeps the tree structure unchanged but needs a shuffling of node entries. Each borrowing needs one extra read and three extra writes. Thus, the cost is at least $(H+5)$. If no siblings can lend an item, we apply the {\it merge} strategy where a node is merged with its sibling. Merging may trigger a cascade of merges bottom-up that involves at most $(2H-1)$ reads and $(2H-2+1)$ writes, in total $(4H-2)$ I/Os. Thus, the fluctuation of the \bpt{} deletion is $\Phi(\text{delete})=(4H-2)-(H+1)=3H-3$ and it is not bounded.
\end{sloppypar}

\section{The Deterministic Insert}\label{sec:algo}
In this section, we present a \bpt{} insert algorithm that guarantees fluctuation-bounded inserts. First, we  introduce the needed definitions in Section~\ref{subsec:algo-def}, then we describe the algorithm in Section~\ref{subsec:algo-insert}, prove its correctness in Section~\ref{subsec:algo-analysis}. We explain how the algorithm is implemented in Section~\ref{subsec:ins-implement} and analyze its cost in Section~\ref{subsec:algo-cost}. 

\subsection{Definitions}\label{subsec:algo-def}

For a node $N$, let $\occ(N)$ be its current occupancy and
$\capf(N)$ its capacity.  For a leaf, occupancy is the number of entries.  For an internal node, occupancy is the number of child pointers.  We define the free space of node $N$ as 
\begin{equation}
    \F(N)=\capf(N)-\occ(N)
\end{equation}

We give the intuition of {\it safe}, {\it unsafe}, and {\it critical} nodes first, then define them formally.
A non-leaf node $N$ is said to be \textit{safe} if, for any next insert into the subtree rooted at $N$, the resulting state of $N$ remains within its capacity constraints and no split is required at $N$. A non-leaf node $N$ is said to be \textit{unsafe} if there exists some next insert into $N$'s subtree that would trigger a split at $N$ (possibly as part of a split propagating from below). We say a non-leaf node is {\em critical} when it is about to transition from being safe to being unsafe. 
Splitting a critical node in a proper way can prevent split propagation from taking place. 

\begin{definition}\label{def:safe-leaf}
    ({\bf Safe and Critical Leaf Nodes}. ) A leaf node $L$ is within capacity if $\F(L)\ge 0$. Among within-capacity leaves, $L$ is {\it safe non-critical} if $\F(L)>0$ and {\it critical} if $\F(L)=0$.
\end{definition}

For an internal node $N$, let $\children(N)$ be its child set, $\CritSet(N)=\{C\in\children(N): C \text{ is critical}\}$ and $\UnsafeSet(N)=\{C\in\children(N): C \text{ is unsafe}\}$. Preallocating slots for critical and unsafe children in their parents could buffer the future splits and prevent split propagation. The children that require reserved parent slots are the union $\NeedSet(N)=\CritSet(N)\cup\UnsafeSet(N)$ and let $\Sfun(N)=|\NeedSet(N)|$.

\begin{definition}\label{def:safe-nonleaf}
    ({\bf Safe, Critical and Unsafe Non-Leaf Nodes}. ) An internal node $N$ is safe if 
    \begin{equation}
        \F(N)\ge \Sfun(N),
    \end{equation}
    $N$ is critical if equality holds, and $N$ is unsafe if
    \begin{equation}
        \F(N)< \Sfun(N).
    \end{equation}
    An internal node is safe non-critical if $\F(N)>\Sfun(N)$.
\end{definition}

In other words, certain nodes are marked critical after an insert. These critical nodes are scheduled for splitting, and with a carefully arranged splitting order, split propagation can be completely avoided. The detailed procedure is described in Section~\ref{subsec:algo-insert}. As will be demonstrated in the correctness proof, the key observation is that it is guaranteed there will be at most one split from any path from root to leaf at any time. Only the critical node closest to the leaf level will be split and this split never propagates upward nor invalidates the tree, and hence ensuring deterministic performance, or bounded-fluctuation insert performance. 

\begin{figure}[h]
    \begin{subfigure}{0.49\linewidth}
        \centering
        \includegraphics[width=\linewidth]{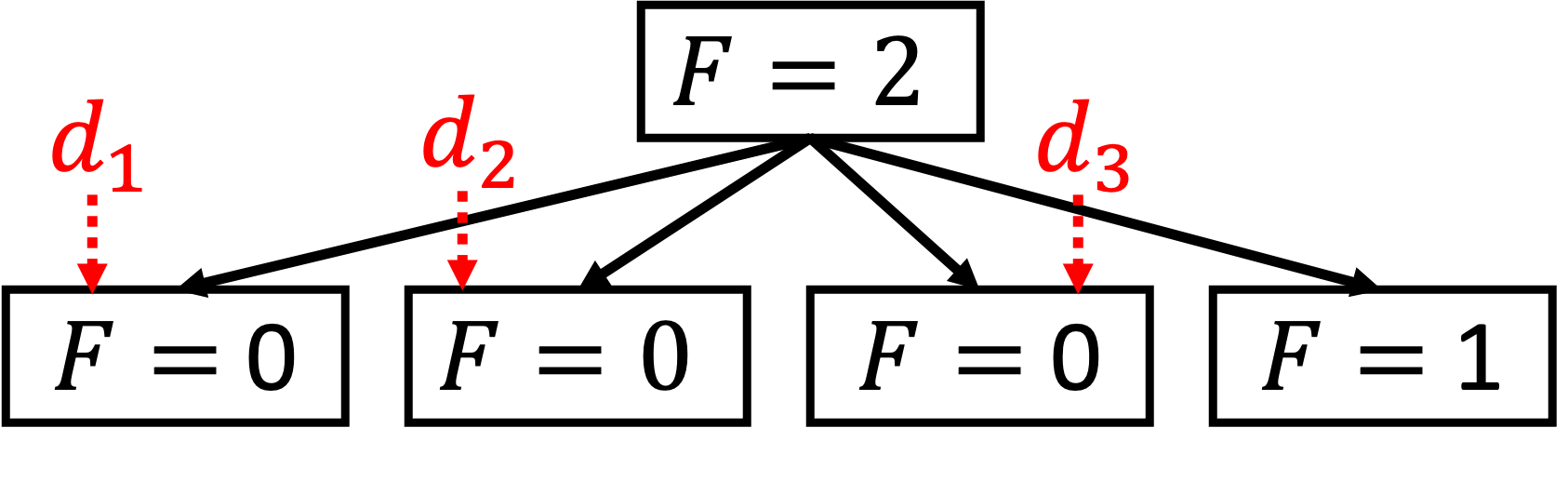}    \caption{}\label{fig:critical_2_level-1}
    \end{subfigure}
    \hfill
    \begin{subfigure}{0.49\linewidth}
        \centering
     \includegraphics[width=\linewidth]{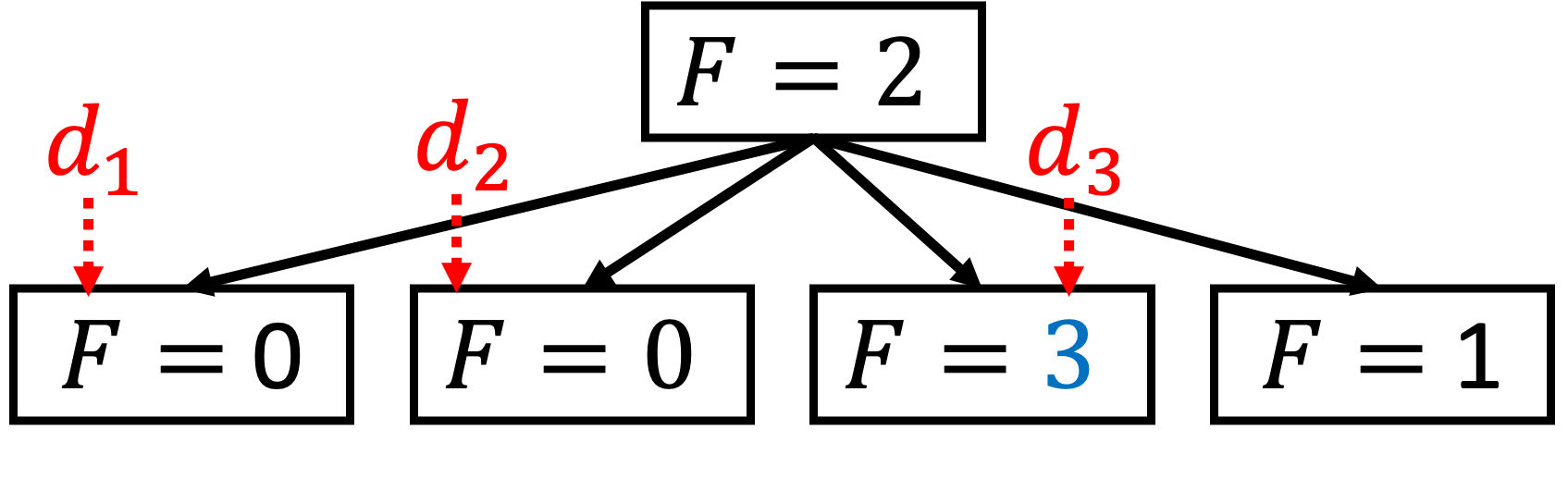}
     \caption{}\label{fig:critical_2_level-2}
    \end{subfigure}
    \vfill
    \begin{subfigure}{0.52\linewidth}
        \centering
        \includegraphics[width=\linewidth]{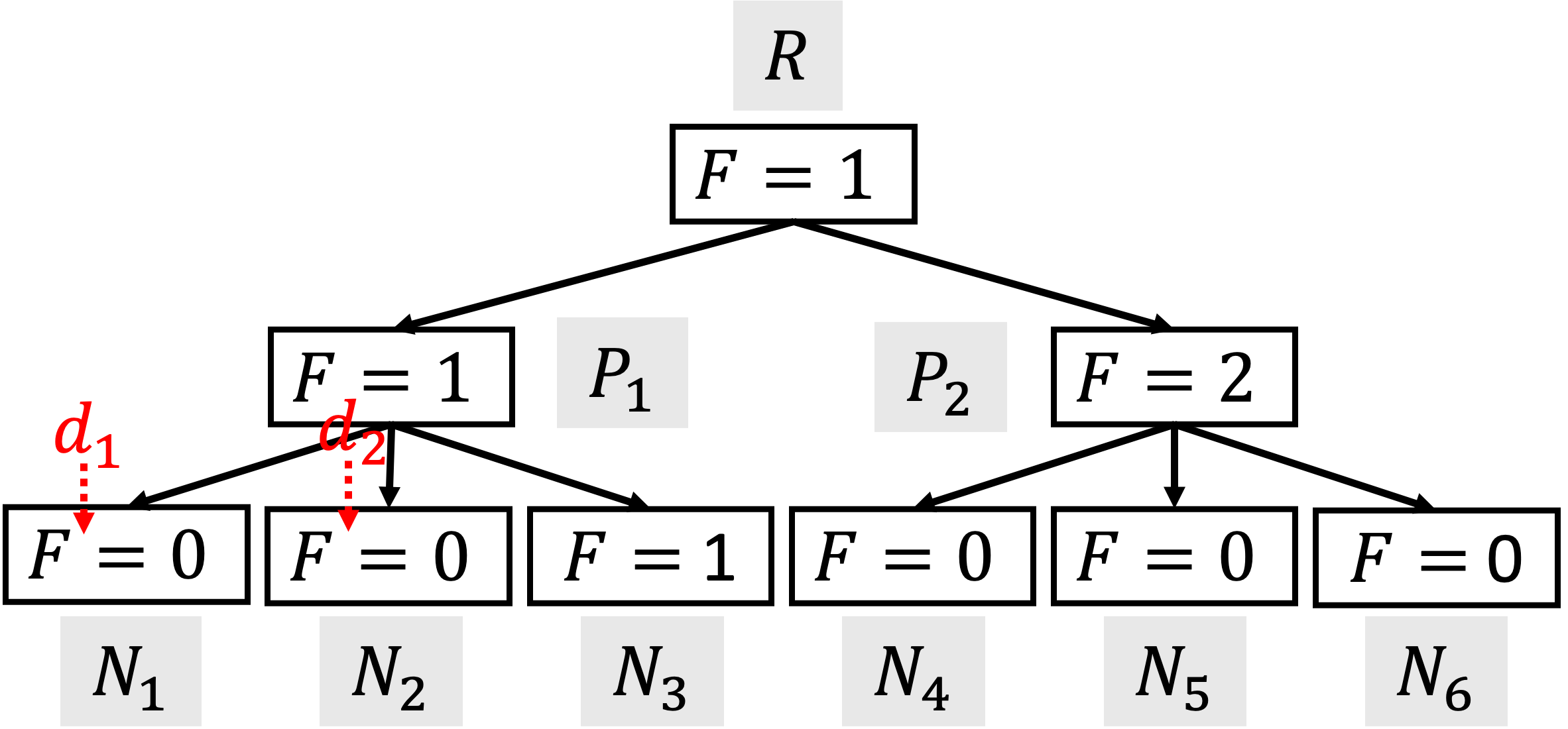}
    \caption{}\label{fig:critical_3_level-1}
    \end{subfigure}
    \hfill
    \begin{subfigure}{0.47\linewidth}
        \centering
        \includegraphics[width=\linewidth]{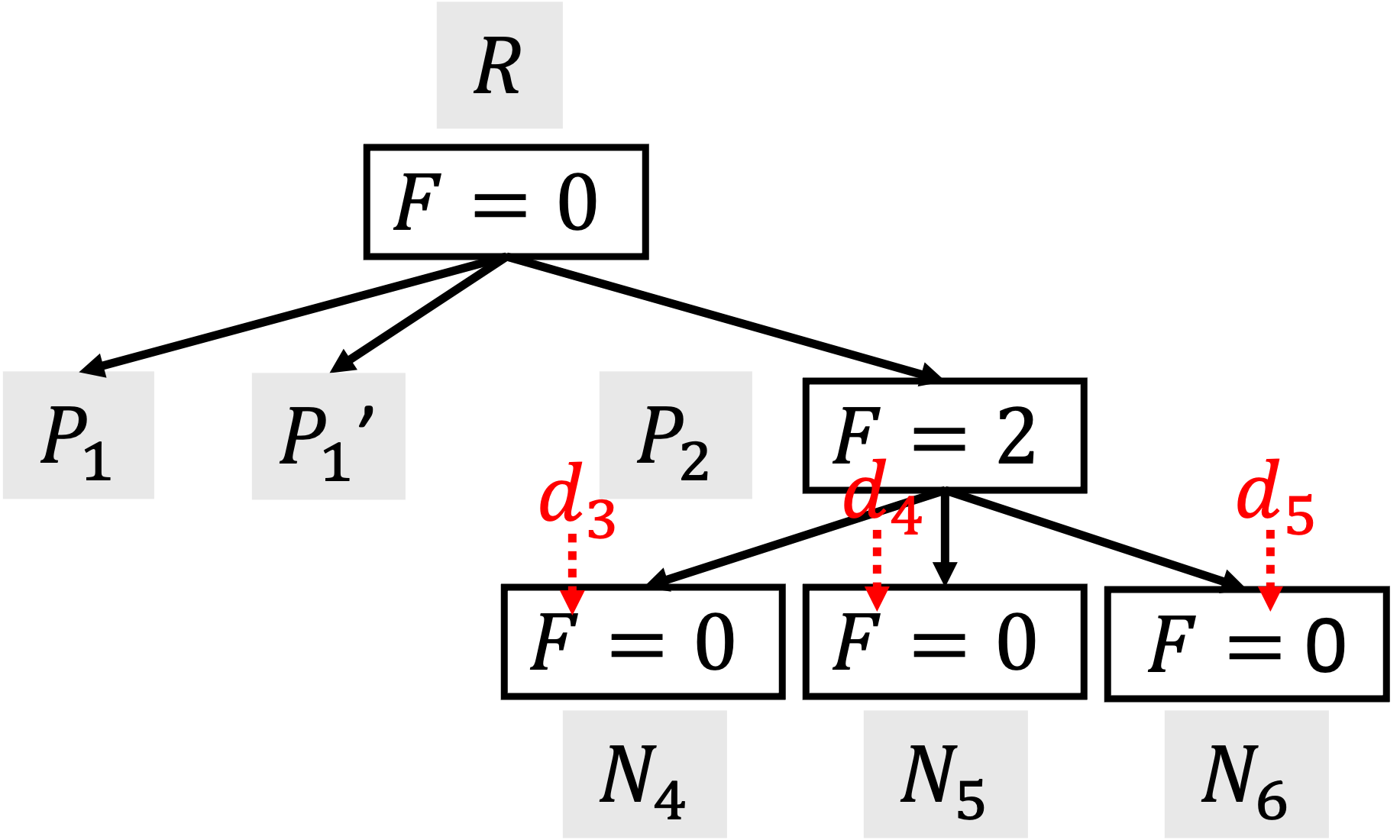}
    \caption{}\label{fig:critical_3_level-2}
    \end{subfigure}
    \caption{Illustrations of critical/safe nodes in    trees of Heights 2 and 3.}
    \label{fig:critical}
\end{figure}

\subsubsection*{Example 1}
Figure~\ref{fig:critical} shows examples of leaf and nonleaf critical nodes. In Figure~\ref{fig:critical_2_level-1}, there are three full leaf nodes and the root node has two free slots. Based on Definition~\ref{def:safe-nonleaf}, the root node is unsafe. There will be split propagation if the next inserts are $d_1$, $d_2$ and $d_3$. However, in Figure~\ref{fig:critical_2_level-2},  $\F(\text{root})=2$, hence, the root node is critical. With the insertion of  $d_1$, $d_2$ and $d_3$, there cannot be a split propagation and it is safe to preemptively split the root node when $d_3$ or later data items are inserted (detailed algorithm in Section~\ref{subsec:algo-insert}).

\subsubsection*{Example 2}
Figure~\ref{fig:critical_3_level-1} shows a tree of Height 3. $P_1$ is unsafe because $N_1$ and $N_2$ are critical and $\F(P_1)=1<2$. $P_2$ is also unsafe. For the root node, since both its children nodes are unsafe, and it has only one empty slot, it is also unsafe. Insertions $d_1$ and $d_2$ can trigger a split propagation from $N_2$ to $P_1$. The tree after split propagation is shown in Figure~\ref{fig:critical_3_level-2}. Continuous inserts of $d_3$, $d_4$ and $d_5$ trigger another split propagation from $N_6$ to $P_2$ and will lead to the root split. This split propagation results in three node splits in a row, i.e., $N_6, P_2, R$.

\subsection{Insertion Algorithm}\label{subsec:algo-insert}
The overview of the insertion process is given in Algorithm~\ref{alg:ins_algo}. We name the tree that adopts this algorithm the \sys{} (\underline{F}luctuation \underline{F}ree B-tree). Insertion is similar to the normal \bpt{} insertion in that the destination leaf node is found by comparing the inserted key with the keys in the non-leaf nodes. The difference lies in that the \sys{} is guaranteed that at most one split takes place from root to leaf per insertion. We will show that no filled nodes are for an insertion.

\begin{algorithm}[!t]
\caption{Insertion ($k$, $v$)}\label{alg:ins_algo}
\begin{minipage}{\columnwidth}

\LinesNumbered

\KwData{$k$, $v$}
$cur = \rootnode$

$last = \Nil$

\While{$cur$ is not a leaf node}{
    \If{$cur$ is critical}{
        $last = cur$
    }
    $cur = \FindNext(cur, k)$
}
\If{$cur$ is critical}{
    $last = cur$
}
\If{$last \neq \Nil $}{
    $\Split(last)$
    
    Insert the new separator and child to the parent
}
$\Insert(cur, k, v)$

Update the critical metadata
\end{minipage}
\end{algorithm}

Starting from the root node, we find the critical node that is closest to the leaf on the path. If there is such a critical node, that critical node is split (we name it proactive split, we also use the term preemptive split interchangeably) and the new split node is inserted to its parent node (line 11) before the leaf node insertion (line 12). Finally, the critical information is updated (details in Section~\ref{subsec:ins-implement}). It is guaranteed that if we always split the critical node that is closest to the leaf level, being a leaf node or a non-leaf node, the parent node of the critical node cannot be full by the time of the preemptive split. Thus, there cannot be any split propagation.

\subsection{Correctness Proof}\label{subsec:algo-analysis}

It is not obvious to see that preemptively splitting the  critical node that is closest to the leaf guarantees only one split per insertion operation. This one split per insertion makes the \sys{} insert deterministically. 
We prove that Algorithm~\ref{alg:ins_algo} preserves the following invariant after every insertion; by induction over the insertion sequence, the invariant implies that at most one split occurs per insertion.

\begin{sloppypar}
Intuitively, this is equivalent of proving the following invariant
\begin{equation}\label{eq:invariant}
    \F(N) \geq \Sfun(N)
\end{equation}
holds for all the non-leaf nodes. We assume that Inequality~\ref{eq:invariant} holds initially. In particular, it holds for an empty tree. We show below that every insertion preserves the invariant; hence, by induction, it holds throughout any sequence of insertions. If so, every split can be accommodated in a non-filled parent node. This invariant is enforced by marking nodes as critical one step earlier than the point where a split would be forced and by splitting only the bottommost critical node. As a result, the split consumes the available slot in the parent but never cascades upward. The proofs below formalize this intuition by showing that unsafe nodes cannot appear and that critical parents can always absorb child splits without becoming unsafe. A useful way to interpret the invariant is that the \sys{} converts a potentially unbounded split cascade into a bounded, localized operation. The bottommost critical node is the only place where a split is allowed to occur, and all higher nodes are guaranteed to have just enough slack to accept the separator. This structure of the \sys{} ensures that the worst-case I/O is dominated by two extra writes at the split level and one at the parent level rather than a chain of writes across levels. In other words, the \sys{} does not eliminate splits, but it {\it deterministically schedules} them so their cost is constant per insertion.
\end{sloppypar}

\begin{lemma}\label{lem:reserved-parent-slot}
Let $P$ be a safe internal node and let $Q$ be a critical child of $P$.  Then $P$ has at least one free slot: $\F(P)\ge 1$.
\end{lemma}

\begin{proof}
Since $Q$ is a critical child of $P$, we have $Q\in\CritSet(P)\subseteq\NeedSet(P)$, so $\Sfun(P)\ge 1$.
Since $P$ is safe, $\F(P)\ge \Sfun(P)$.  Thus $\F(P)\ge 1$.
\end{proof}



For the proofs below, define the \emph{slack} of an internal node $N$ as
\begin{equation}\label{eq:def_delta}
    \Delta(N)=\F(N)-\Sfun(N).
\end{equation}

We call a split of a node $Q$ into $Q_\ell$ and $Q_r$ a \emph{safe split} if
\[
    \Delta(Q_j)\geq2,
    \qquad Q_j\in\{Q_\ell,Q_r\}.
\]
The following lemma shows that splitting a critical node evenly is a safe split.

\begin{lemma}[Slack after splitting a critical node]\label{lem:split-slack}
Let $Q$ be a critical node of capacity $M\geq 6$, and let a balanced split partition $Q$ into $Q_\ell$ and $Q_r$. Then that split is a safe split and each resulting node $Q_j\in\{Q_\ell,Q_r\}$ satisfies $\Delta(Q_j)\geq 2$.
\end{lemma}

\begin{proof}
Let $n=\occ(Q)$ and $s=\Sfun(Q)$. Since $Q$ is critical, $\F(Q)=\Sfun(Q)=s$. Since $\F(Q)=M-n$, we have $n+s=M$.

\noindent
If $Q$ is a leaf node, then $s=0$, and we have
\[
\begin{aligned}
    \Delta(Q_j)
      &= F(Q_j)\\
      &= M-\occ(Q_j)\\
      &\geq M-\lceil n/2\rceil\\
      &= \lfloor M/2\rfloor\geq 3 > 2.
\end{aligned}
\]
Therefore, the split of a leaf node is a safe split.

\noindent
If $Q$ is an inner node, $\Sfun(Q_j)$ must satisfy the following two inequalities:
\[
    \Sfun(Q_j)\leq \Sfun(Q) = s,
    \qquad \Sfun(Q_j)\leq \lceil n/2 \rceil.
\]
Thus $\Sfun(Q_j)\leq \text{min} (s, \lceil n/2 \rceil)$ and 
\[
\begin{aligned}
    \Delta(Q_j)
      &= F(Q_j)-\Sfun(Q_j)\\
      &\geq F(Q_j)-\text{min} (s, \lceil n/2\rceil)\\
      &\geq M-\lceil n/2\rceil-\min(s,\lceil n/2\rceil).
\end{aligned}
\]


When $s\leq \lceil n/2\rceil$, i.e., $s\leq M/3$,
\[
\begin{aligned}
    \Delta(Q_j)
      &\geq M-\lceil n/2\rceil-\text{min} (s, \lceil n/2\rceil)\\
      &= M-\lceil (M-s)/2\rceil-s\\
      &\geq \lfloor M/2 \rfloor -s/2\\
      &\geq M/3 \geq 2.
\end{aligned}
\]
When $s> \lceil n/2\rceil$, i.e., $s> M/3$,
\[
\begin{aligned}
    \Delta(Q_j)
      &\geq M-\lceil n/2\rceil-\text{min} (s, \lceil n/2\rceil)\\
      &= M-\lceil n/2\rceil-\lceil n/2\rceil\\
      &\geq M-n-1 = M-(M-s)-1\\
      &= s-1 \\
      &> M/3-1 \geq 1.
\end{aligned}
\]

Thus, in both cases, $\Delta(Q_j)\geq2$ for $Q_j\in\{Q_\ell,Q_r\}$ holds.
\end{proof}

\begin{lemma}[Safety of splitting the bottommost critical node]\label{lem:critical-split}
Let $Q$ be the bottommost critical node on the insertion path, and let a safe proactive split step split $Q$ into $Q_\ell$ and $Q_r$. 
Then both $Q_\ell$ and $Q_r$ are safe non-critical immediately after the split. Moreover, the split result whose subtree receives the pending insertion remains safe non-critical after the insertion.
\end{lemma}

\begin{proof}
Since $Q$ is critical before the proactive split, $\F(Q)=\Sfun(Q)$. After the balanced split, the entries/children of $Q$ are partitioned between $Q_\ell$ and $Q_r$. According to safe split, for $Q_j\in\{Q_\ell,Q_r\}$, $\Delta(Q_j) = \F(Q_j)-\Sfun(Q_j) \ge 2 > 0$. Hence $\F(Q_j)>\Sfun(Q_j)$, so both $Q_\ell$ and $Q_r$ are safe non-critical immediately after the split.

Let $Q_*$ be the split result whose subtree receives the pending insertion. We distinguish whether $Q_*$ is a leaf or an internal node. If $Q_*$ is a leaf, the pending insertion consumes one free entry, so
\[
   \F'(Q_*)=\F(Q_*)-1.
\]
Since a leaf has no children, $\Sfun(Q_*)=\Sfun'(Q_*)=0$. Hence
\[
   \Delta'(Q_*)=\Delta(Q_*)-1\ge 1.
\]

If $Q_*$ is an internal node, the pending insertion occurs in a leaf in the subtree rooted at $Q_*$. Because $Q$ was the bottommost critical node before the proactive split, every node below $Q$ on the insertion path was safe non-critical and the destination leaf was not full. Therefore, no further split occurs below $Q_*$. Consequently, the number of child pointers of $Q_*$ does not change, and 
\[
   \F'(Q_*)=\F(Q_*).
\]
Only the child of $Q_*$ on the insertion path can change its status and contribute to $\Sfun(Q_*)$. Thus
\[
   \Sfun'(Q_*)\le \Sfun(Q_*)+1,
\]
and therefore
\[
\begin{aligned}
    \Delta'(Q_*)
      &=\F'(Q_*)-\Sfun'(Q_*)\\
      &\ge\F(Q_*)-\bigl(\Sfun(Q_*)+1\bigr)\\
      &=\Delta(Q_*)-1\\
      &\ge1.
\end{aligned}
\]

Thus, in either case, $\Delta'(Q_*)\ge1$. Hence the split result receiving the pending insertion remains safe non-critical. The other split result is unchanged and remains safe non-critical.
\end{proof}

\begin{lemma}\label{lem:split-child-parent-safe}
Let $P$ be a safe internal node and let $Q$ be a critical child of $P$. Suppose $Q$ is split into $Q_\ell$ and $Q_r$ by a safe split. After inserting the new child pointer into $P$ and completing the pending insertion below one of $Q_\ell$ and $Q_r$, the safety status of $P$ is unchanged. In particular, $P$ remains safe and no split is forced at $P$.
\end{lemma}

\begin{proof}
Before the split, $Q$ is an element of $\NeedSet(P)$, contributing one unit to $\Sfun(P)=|\NeedSet(P)|$. Because the split is safe and by Lemma~\ref{lem:critical-split}, it replaces $Q$ by two children that are safe non-critical after the pending insertion is completed. Thus, $Q$'s old contribution is removed from $\NeedSet(P)$ and neither replacement contributes a new element to $\NeedSet(P)$.

The parent consumes one free slot to store the new child pointer and separator. Hence both sides of Inequality~\ref{eq:invariant} decrease by one:
\[
  \F'(P)=\F(P)-1,\qquad \Sfun'(P)=\Sfun(P)-1.
\]
Therefore $\F'(P)-\Sfun'(P)=\F(P)-\Sfun(P)$. The parent's slack remains unchanged, and hence its safe non-critical or critical status, is unchanged. Since $P$ was safe before the split, $P$ is safe afterwards.
\end{proof}

\begin{lemma}\label{lem:safe-noncritical-path}
Consider a root-to-leaf path $N_0,N_1,\ldots,N_h$ in which every internal node $N_i$, $0\le i<h$, is safe non-critical, and the leaf $N_h$ is not full. An ordinary insertion into $N_h$, followed by critical-bit recomputation on this path, preserves Inequality~\ref{eq:invariant} for every node on the path and performs no split.
\end{lemma}

\begin{proof}
The leaf $N_h$ has at least one free entry, so the insertion does not overflow the leaf. The leaf may become critical after the insert. Moving one level upward, this can increase the $\Sfun(N_{h-1})$ of the parent $N_{h-1}$ by at most one. Since the parent was safe non-critical $\F(N_{h-1})>\Sfun(N_{h-1})$, and both sides are integers, increase of $\Sfun(N_{h-1})$ implies
\[
    \F(N_{h-1})\ge \Sfun(N_{h-1})+1.
\]
Therefore the possible one-unit increase still leaves the parent safe.

The same argument repeats inductively upward along the path. 
Suppose $N_{i+1}$ remains safe but changes from safe non-critical to critical. Then, and only then, it adds one new contribution to $\Sfun(N_i)$. Thus 
\[
   \Sfun'(N_i)\le\Sfun(N_i)+1.
\]
Since $N_i$ was safe non-critical before the insertion, 
\[
    \F(N_i)\ge\Sfun(N_i)+1\ge\Sfun'(N_i).
\]
The insertion does not change the number of child pointers of $N_i$, so $\F'(N_i)=\F(N_i)$. Therefore $N_i$ remains safe. Applying this argument inductively from $N_{h-1}$ to $N_0$ proves the claim.
\end{proof}

\begin{lemma}\label{lem:no-critical-signal-crosses-split}
Let $Q$ be the bottommost critical node on the insertion path. After the safe split of $Q$ and the pending insertion, no ancestor of $Q$ gains a new critical or unsafe child contribution.
\end{lemma}

\begin{proof}
If $Q$ is the root, it has no ancestor and the claim is immediate. Otherwise, let $P$ be the parent of $Q$. By Lemma~\ref{lem:split-child-parent-safe}, replacing the critical child $Q$ by the two safe split results leaves the safety status of $P$ unchanged after the pending insertion. In particular, $P$ is critical after the operation if and only if it was critical before the operation, and the same is true for being safe non-critical.

Now consider the parent of $P$, $R$ if it exists. $R$'s child on the insertion path is still $P$, and $P$'s critical/unsafe contribution has not changed. Therefore $R$ does not gain a new contribution to its $\Sfun(R)$. Repeating the same argument upward proves that no strict ancestor of $Q$ gains a new critical or unsafe child contribution.
\end{proof}

\begin{lemma}\label{lem:no-unsafe-reached}
Assume the Invariant~\ref{eq:invariant} holds before an insertion. After one insertion by the \sys{} insertion algorithm, no internal node can transition into the unsafe state.
\end{lemma}

\begin{proof}
Let $p$ be the root-to-leaf search path for the inserted key.

First suppose that $p$ contains no critical node. Then every internal node on $p$ is safe non-critical and the destination leaf is not full. No proactive split is performed. By Lemma~\ref{lem:safe-noncritical-path}, the insertion and metadata recomputation preserve Inequality~\ref{eq:invariant} on the path. Nodes outside $p$ are unchanged, so no internal node becomes
unsafe.

Now suppose that $p$ contains at least one critical node, and let $Q$ be the bottommost critical node on $p$. Algorithm~\ref{alg:ins_algo} splits $Q$ and no lower node on the path is critical before the insertion. By Lemma~\ref{lem:critical-split}, the split results are safe non-critical, and the split result that receives the pending insertion remains safe non-critical after that insertion. The nodes on $p$ that are below $Q$ contain no pre-existing critical node, because $Q$ was chosen bottommost. Moreover, the destination leaf is not full; otherwise that leaf itself would be critical and would have been selected instead of $Q$. Therefore, Lemma~\ref{lem:safe-noncritical-path} preserves the invariant below $Q$.

It remains to consider ancestors of $Q$. If $Q$ has a parent $P$, then by the invariant $P$ is safe before the operation. By Lemma~\ref{lem:split-child-parent-safe}, $P$ remains safe after it absorbs the new separator and child pointer. By Lemma~\ref{lem:no-critical-signal-crosses-split}, no strict ancestor of $Q$ gains a new critical or unsafe child contribution. Since their free space is unchanged, they preserve Inequality~\ref{eq:invariant}. Nodes not on the search path $p$ are unchanged. Hence no internal node can become unsafe.

If $Q$ is the root, $\Split(Q)$ creates a new root whose children are $Q_\ell$ and $Q_r$. By Lemma~\ref{lem:critical-split}, both children are safe non-critical after completion of the pending insertion. Hence neither child contributes to the new root's $\Sfun$. The new root therefore satisfies $\F(root)\ge\Sfun(root)=0$, so Inequality~\ref{eq:invariant} also holds after a root split.
\end{proof}

\begin{theorem}\label{thm:no-split-propagation}
For every insertion into a \sys{} satisfying Inequality~\ref{eq:invariant}, Algorithm~\ref{alg:ins_algo} performs at most one split. Therefore split propagation cannot occur.
\end{theorem}

\begin{proof}
Consider one insertion. If its search path contains no critical node, Algorithm~\ref{alg:ins_algo} performs no proactive split, and Lemma~\ref{lem:safe-noncritical-path} shows that the pending insertion causes no split.  Otherwise, let $Q$ be the bottommost critical node. The algorithm performs exactly one proactive split, namely the split of $Q$. By Lemma~\ref{lem:reserved-parent-slot}, if $Q$ is not the root its parent has a free slot for the new child pointer, so this split does not propagate upward. By Lemma~\ref{lem:critical-split} and Lemma~\ref{lem:safe-noncritical-path}, completing the pending insertion causes no additional split below $Q$. Finally, Lemma~\ref{lem:no-unsafe-reached} shows that Inequality~\ref{eq:invariant} is preserved for the next insertion. Therefore Algorithm~\ref{alg:ins_algo} performs at most one split per insertion, and split propagation cannot occur.
\end{proof}

\subsection{Implementation}\label{subsec:ins-implement}
We presented and proved the correctness of the insertion algorithm of \sys{} above. Recording whether the node's status, i.e., being critical or not, becomes crucial. What's more, such information should be maintained properly when the insertion comes.

We can use one bit to denote whether a node is critical or not. If such bit is stored within the node, changing the node's status requires at least one I/O to write. In situations where a change of status may cause a chain or propagation of changes of the ancestors, unexpected I/Os may disrupt the local splits of the \sys{}.

Therefore, we use an in-memory bitmap to store the nodes' status. The node ID is used to locate the bit. Reading and setting the bit requires only bit manipulation. The size of that bitmap is minimal. For a \bpt{} that has a fanout of 500 and height of 4 with 4kB pages, the bitmap takes only 15MB of memory, which is negligible compared to the index size that is around 500GB.

\subsection{Cost Analysis}\label{subsec:algo-cost}

The cost of the \sys{} insertion includes a single root-to-leaf traversal as well as the proactive split, plus finishing the pending insertion. The top-down traversal requires $H$ node reads for a Height-$H$ tree. The proactive split includes splitting the critical node, i.e., two nodes writes, and installing the new node to the parent node, i.e., one node write, thus in total three node writes. Finishing the pending insertion needs one write into the destination leaf node. Maintaining the nodes' status requires no extra I/Os. Summing together results in the insertion cost of $(H + 4)$. Thus, the worst-case I/O differs from the best case by only three writes, yielding {\it fluctuation-bounded} behavior under our definition while preserving the same asymptotic search cost as a standard \bpt{}.

Let us analyze whether our insertion algorithm is a de-amortization one. First, we investigate the cost of a de-amortized \bpt{} insertion algorithm. As the cost of reading a node is significantly smaller than writing a node, we only consider writes in this analysis. Suppose there are $n$ insertions into a tree of height $H$. The total cost of $n$ insertions is $f(n)=O(n)$ because one insert modifies one node. Therefore, each insert operation has at most $O(1)$ budget to finish. \bpt{} insert takes at most $2H+1$ which is $O(\log n)$ and this is not an de-amortization algorithm. On the other hand, the \sys{} takes at most cost of 4 to finish one insert, which makes the algorithm a de-amortized one.

\section{The Deterministic Deletion}\label{sec:deletion}
We present a deterministic insertion algorithm for the \sys{} that guarantees the cost of any insertion to be at most $(H+4)$ I/Os. Next, we show how the \sys{} handles the deletes. The deletion algorithm is described in Section~\ref{subsec:del-algo}; the proof and the analysis are included in Sections~\ref{subsec:del-proof} and~\ref{subsec:del-cost}. 

\subsection{Definitions}

Let $\U(N)$ be the utilization ratio of node $N$. It is computed as $\U(N)=\occ(N)/\capf(N)$. By the definition of the \bpt{}, a non-root node $N$ {\it underflows} when its occupancy $\occ(N) < 50\% \times\capf(N)$, or $\U(N)<50\%$. An underflowing node either {\it borrows} an entry from its immediate non-underflowing sibling nodes, or {\it merges} with its half-full sibling node. The merge operation is the opposite operation of split that can cause propagation and the non-deterministic cost.

\subsection{Deletion Algorithm}\label{subsec:del-algo}

The deletion algorithm is described in Algorithm~\ref{alg:del_algo}. The algorithm is almost identical to the deletion algorithm of the \bpt{} except that the temporary underflowing nodes are allowed and only the underflowing node that is closest to the leaf level is handled each time to avoid huge increase in I/O usage. When the underlowing node is able to find a sibling to borrow entry, the utilization decrease does not propagate to the parent node. If such a sibling cannot be found, merge operation is triggered. 


\begin{algorithm}[!t]
\caption{Deletion ($k$)}\label{alg:del_algo}
\begin{minipage}{\columnwidth}

\LinesNumbered

\KwData{$k$}
$cur = \rootnode$

$last = \Nil$

\While{$cur$ is not a leaf node}{
    \If{$\U(cur)<50\%$}{
        $last = cur$
    }
    $cur = \FindNext(cur, k)$
}
\If{$\U(cur)<50\%$}{
    $last = cur$
}
\If{$last \neq \Nil $}{
    $sibling=\FindSibling(last)$

    \eIf{$U(sibling)>50\%$}{
        $\Borrow(last, sibling)$
    } {
        $\Merge(last, sibling)$
    }
    
    Update keys of parent
}
$\Delete(cur, k)$

\end{minipage}
\end{algorithm}

\subsection{Correctness Proof}\label{subsec:del-proof}

Given the algorithm~\ref{alg:del_algo}, we prove the following theorem.

\begin{theorem}\label{thm:no-merge-propagation}
For every deletion into a \sys{}, Algorithm~\ref{alg:del_algo} performs at most one merge. Therefore merge propagation cannot occur.
\end{theorem}

\begin{proof}
\begin{sloppypar}
During the root-to-leaf traversal, Algorithm~\ref{alg:del_algo} records the bottommost underflowing node $N$ on the search path. After the traversal, it performs maintenance only on $N$: either $\Borrow(N,sibling)$ or $\Merge(N,sibling)$. No maintenance is recursively performed on $N$'s parent, even if merging $N$ causes its parent to underflow. The data item is then deleted from the target leaf without another structural repair. Therefore, at most one $\Merge$ operation is executed per deletion. 
\end{sloppypar}
\end{proof}

\subsection{Implementation}\label{subsec:del-implementation}

The implementation of Algorithm~\ref{alg:del_algo} is straightforward and does not require any maintenance. We only need to keep track of the bottommost node that is below 50\% that takes only one extra variable.

\subsection{Cost Analysis}\label{subsec:del-cost}

The cost of deletion includes the traversal of the nodes from root to the target leaf, as well as the local merge or borrow. For a Height-$H$ tree, the top down traversal of the number of read nodes is $H$. The bookkeeping does not require any extra I/Os. The local maintenance operation needs the sibling node that uses one I/O to lookup. Borrowing rewrites both the sibling node and the destination node, as well as the parent node. Thus, the borrowing operation uses three extra node writes. Merging rewrites one of the sibling and destination nodes, as well as the parent node. That requires two extra node writes. Summing together results in the deletion cost of $H+5$ for borrowing, or $H+4$ for merging. The fluctuation between the worst case cost and the minimum cost ($H+1$), is constant, i.e., 3.


\section{Simulation Results}

\begin{figure*}[!b]
    \begin{subfigure}{\linewidth}
        \centering
        \includegraphics[width=0.9\linewidth]{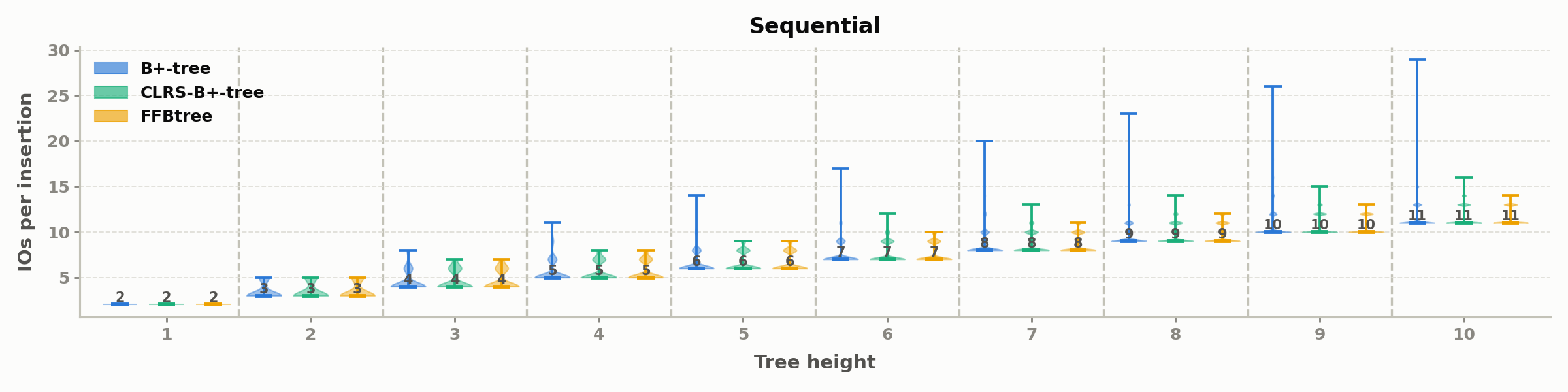}
        \caption{}\label{fig:sim_insert_seq}
    \end{subfigure}
    \vfill
    \begin{subfigure}{\linewidth}
        \centering
        \includegraphics[width=0.9\linewidth]{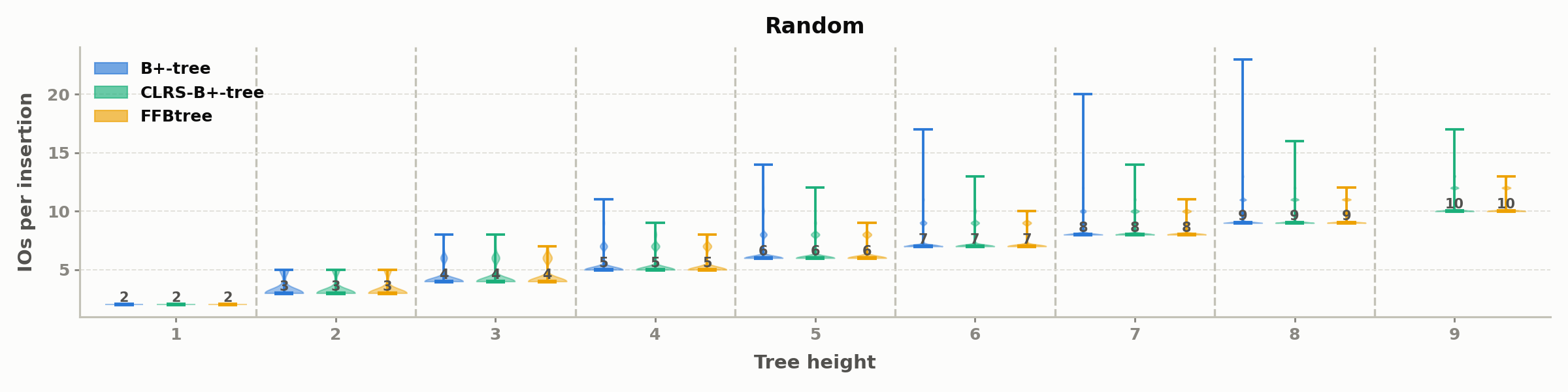}
        \caption{}\label{fig:sim_insert_random}
    \end{subfigure}
    \vfill
    \begin{subfigure}{\linewidth}
        \centering
        \includegraphics[width=0.9\linewidth]{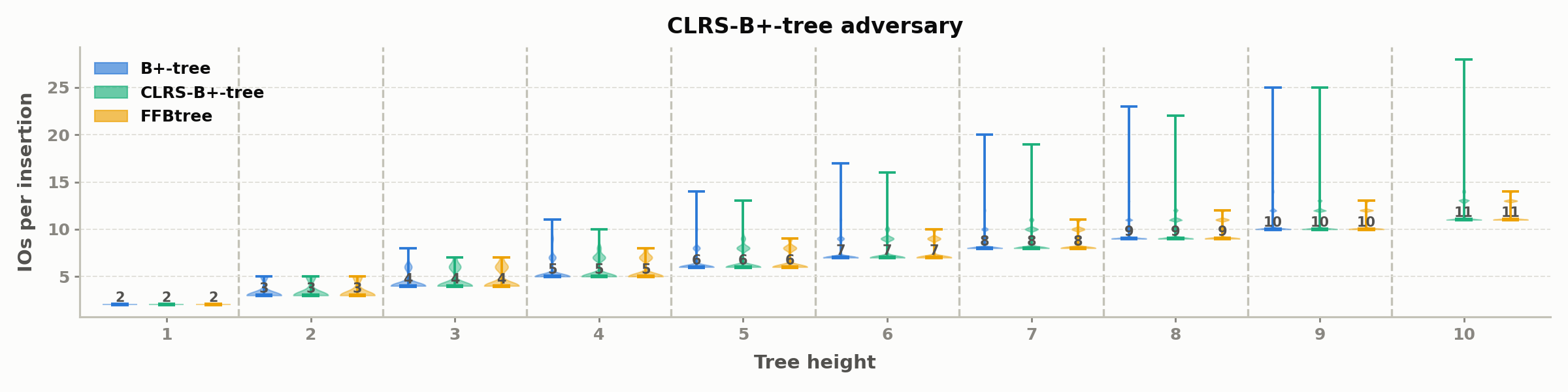}
        \caption{}\label{fig:sim_insert_adv}
    \end{subfigure}
    \caption{Simulated insertions of three indexes under three datasets. Median values are shown in the plot.}
    \label{fig:sim_insert}
\end{figure*}

We compare the \sys{} with two baseline indexes under a simulated environment where nodes track only occupancy counts.

\noindent
{\bf The \bpt{}}: Described in~\cite{bayer1970organization,comer1979ubiquitous}. Inserts first find the target leaf node, insert into the leaf node, and split as needed. Deletes remove the target entries in the leaf node. In the case of an underflowing node, the node finds a sibling that can borrow, if not, the node merges with the sibling.

\noindent
\begin{sloppypar}
{\bf The \clrs{}}: Described in the algorithm textbook~\cite{clrs}. Instead of splitting reactively due to overflowing, nodes are split preemptively in the descending process if they are full. Suppose insertion of $d_1$ does not encounter any full nodes, the total cost is $H+1$. Insertion of $d_2$ encounters one full node from root to leaf, that full node is split before $d_2$ is inserted and the total cost is $(H+1+3)$. Suppose in the worst case, all the nodes along the path are full, the worst-case cost is $(3H+1)$, the same as the \bpt{}. Fluctuation is also the same, which is $\Phi=2H \notin o(H)$, not bounded.
\end{sloppypar}

For the simulated insertion workload, we generate three datasets:

\noindent
{\bf Sequential}: The inserted integer keys are generated in ascending order. This mimics the workload of time-series data, append-only logs, version history etc.

\noindent
{\bf Uniformly random}: The inserted integer keys are equally likely distributed across the entire key space. This mimics the workload of hashed keys, UUID (Universally Unique Identifier) keys etc.

\noindent
{\bf The \clrs{} adversary}: The inserted integer keys are arranged in a sequence such that nodes on the rightmost root-to-leaf path can be filled from time to time. For example, the rightmost leaf node is the first to be filled. Then in order to fill that leaf node's parent, inserts go to that leaf node's siblings until that parent node is filled. Next the grandparent node is filled in the same way. The steps repeat until the root is filled and the rightmost path ends up in all nodes that are filled.

For the simulated deletion workload, we generate two workloads: 

\noindent
{\bf Sequential}: The tree is pre-loaded with integer keys. The deleted keys are generated in ascending order, i.e., the leftmost key is deleted.

\noindent
{\bf Uniformly random}: The tree is pre-loaded with integer keys. The deleted integer keys are equally likely distributed across the entire key space.

We use a small node size (8 entries) to emulate deep trees. We use a 192-core machine with Intel(R) Xeon(R) Platinum 8168 CPUs @ 2.70GHz, 2~TB memory, and Linux 4.15.0. All indexes are implemented in C++. The simulation intentionally removes buffer effects and page content overhead so that the only source of fluctuation is the change of the tree structure. This isolates the algorithmic contribution of the \sys{} and avoids conflating its behavior with cache hit rates or device noise.

\subsection{Simulated Insertion}

We issue one million inserts for each dataset into each index. The results are presented in Figure~\ref{fig:sim_insert}. The $x$-axis is divided into different tree heights. For each tree height, we compare the distribution of the number of I/Os per index. We use violin plot to show the minimum, maximum and median values of each distribution.

\begin{figure}[h]
    \begin{subfigure}{\linewidth}
        \centering
        \includegraphics[width=0.59\linewidth]{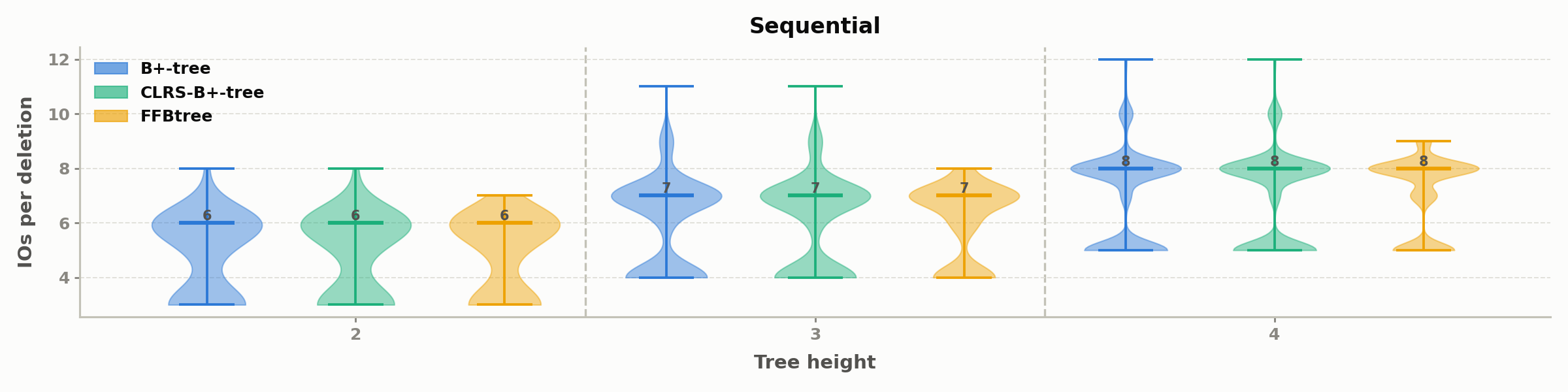}
        \caption{}\label{fig:sim_delete_seq}
    \end{subfigure}
    \vfill
    \begin{subfigure}{\linewidth}
        \centering
        \includegraphics[width=0.59\linewidth]{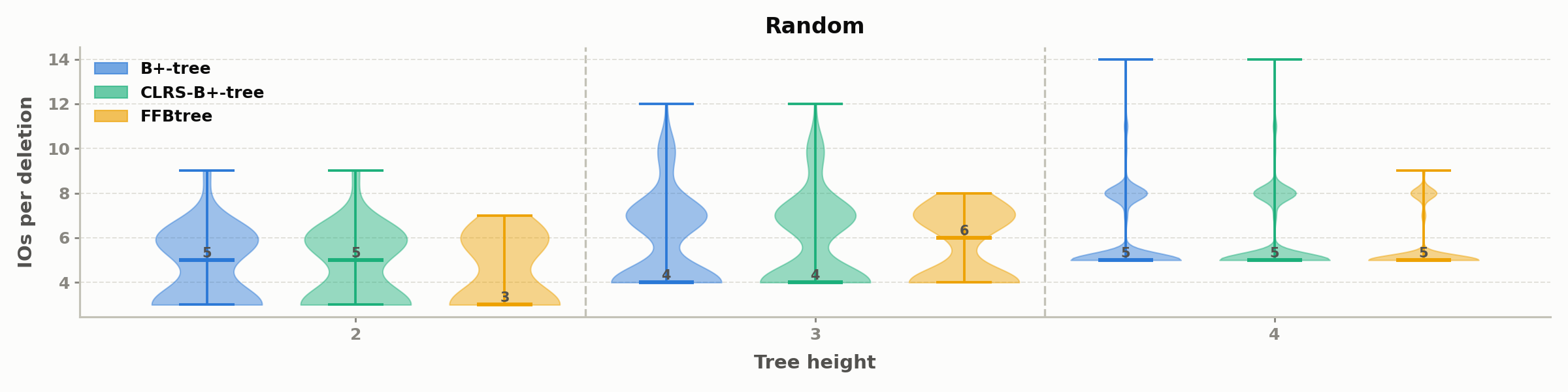}
        \caption{}\label{fig:sim_delete_random}
    \end{subfigure}
    \vfill
    \begin{subfigure}{\linewidth}
        \centering
        \includegraphics[width=0.59\linewidth]{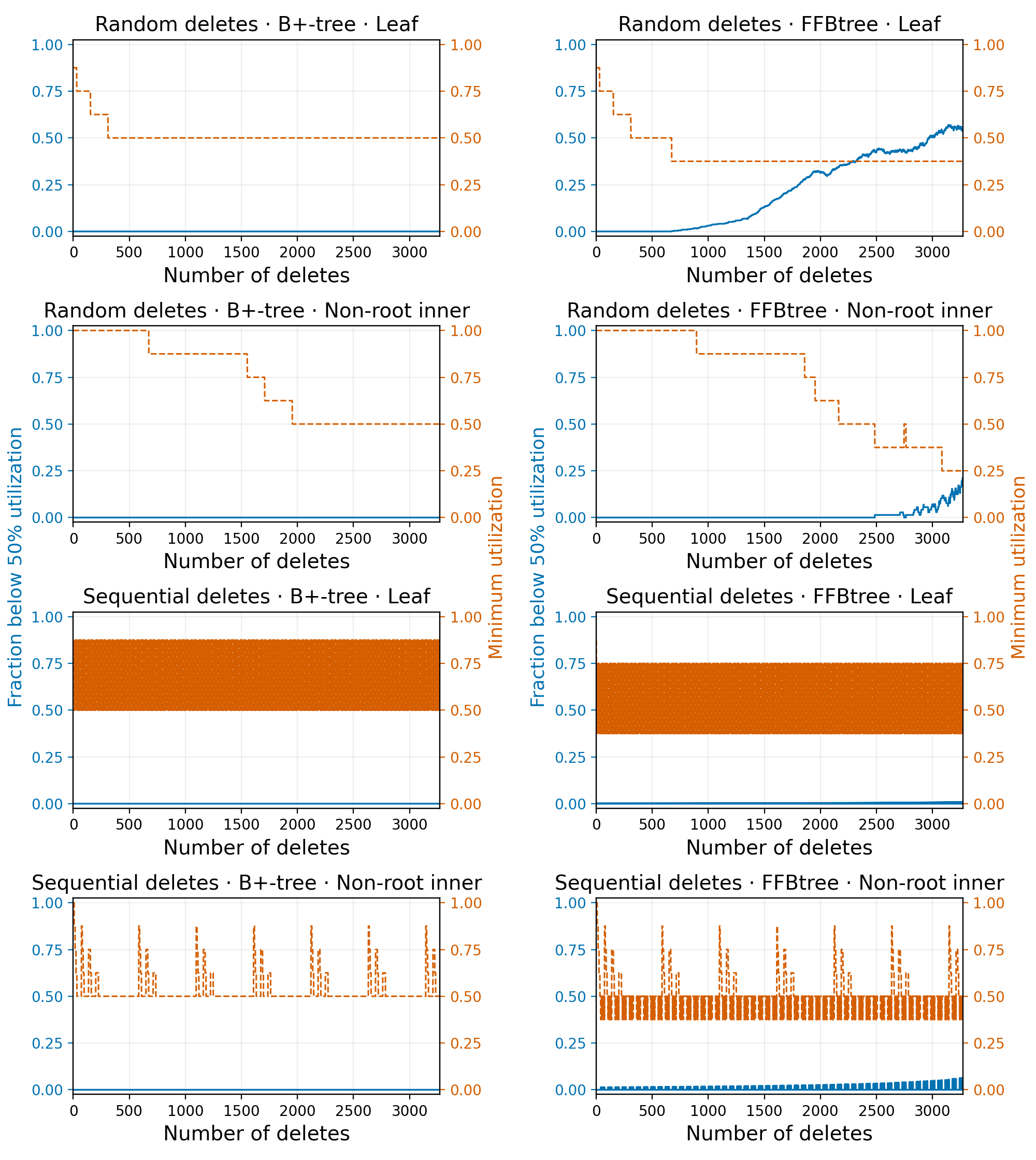}
        \caption{}\label{fig:sim_delete_util}
    \end{subfigure}
    \caption{(a) and (b): Simulated deletions of three indexes under under two datasets. (c) Minimum utilization and fraction of underflowing nodes.}
    \label{fig:sim_delete}
\end{figure}

\begin{figure*}[h]
    \begin{subfigure}{\linewidth}
        \centering
        \includegraphics[width=0.9\linewidth]{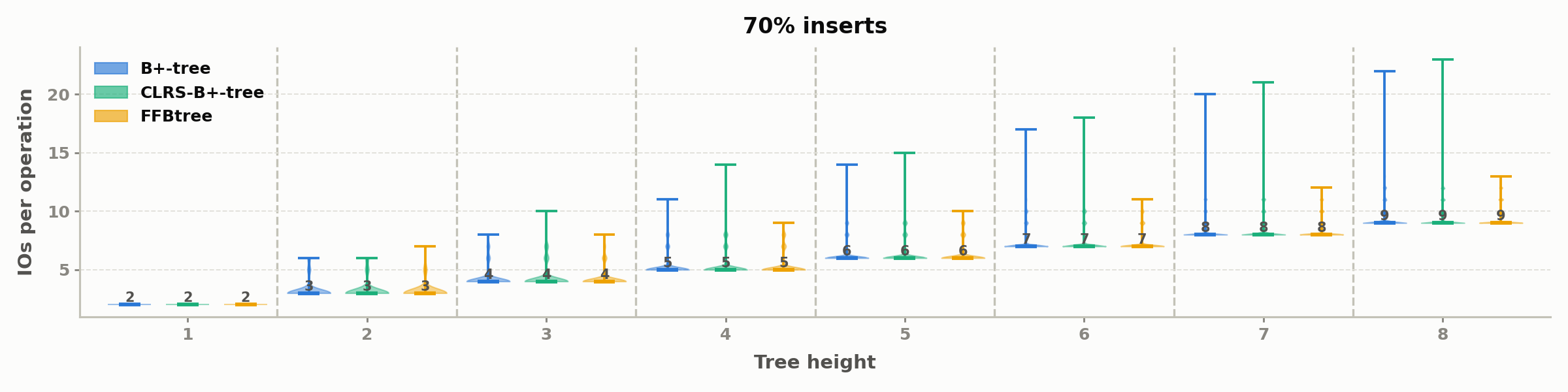}
        \caption{}\label{fig:sim_0.7}
    \end{subfigure}
    \vfill
    \begin{subfigure}{\linewidth}
        \centering
        \includegraphics[width=0.9\linewidth]{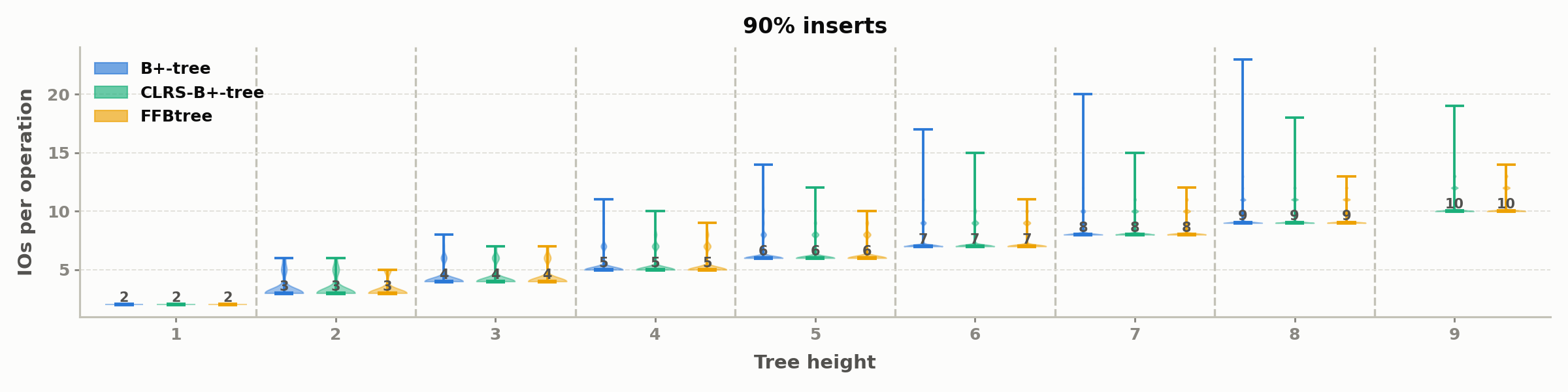}
        \caption{}\label{fig:sim_0.9}
    \end{subfigure}
    \caption{Simulated mixed workload of three indexes.}
    \label{fig:sim_mix}
\end{figure*}

When the height of the tree is small ($H=2$), the three indexes do not exhibit any difference among the maximum and minimum number of I/Os. As the tree grows deeper, the fluctuation of the \bpt{} begins to increase rapidly, from 0 ($H=1$) to 18 ($H=10$) in the sequential dataset, 0 ($H=1$) to 14 ($H=8$) in the uniform random one, 0 ($H=1$) to 15 ($H=15$) in the \clrs{} adversary one. While for the \clrs{}, this fluctuation growth is much smaller in the sequential dataset: 0 ($H=1$) to 5 ($H=10$), slightly increased in uniform random dataset: 0 ($H=1$) to 7 ($H=9$), largest in the \clrs{} adversary: 0 ($H=1$) to 17 ($H=10$) as that one is designed to make the \clrs{} suffer. The \sys{} remains to have the lowest and stable fluctuation of all. In fact, the fluctuation of the \sys{} is constant that is always 3 under all tested datasets.

Observe that under sequential dataset, the \bpt{} has the maximum I/O and fluctuation much larger than the \clrs{} and the \sys{} when the tree height is greater than four. Actually, the sequential dataset is one of the adversary datasets for the \bpt{}. When data keep inserting to the rightmost leaf node, its parent node becomes full as a result. As the immediate ancestor fills, the split of the ancestor causes the parent of the immediate ancestor to overflow. The fullness of the node gradually passes towards the root until all the nodes along the rightmost path become full and one insert makes the split propagate to the root. 

\begin{sloppypar}
Under sequential dataset, the \clrs{} has the same maximum I/O as the \sys{} when the tree height is 5 or below. This implies both two indexes have at most one node split per insertion. The \clrs{} starts to have more node splits when the tree height continues to grow. Note that the inner node is full only after one of its children splits and that inner node needs to wait until next insertion to be split. Thus, that node lags its child node by one insert and so are other adjacent levels. Since the simulated node size is eight, leaf node has the same cycle as the node whose depth is four. In a real \clrs{} whose fanout is usually hundreds of keys, it is hard to find a tree with height in hundred. Therefore, we can regard the \clrs{} having the same maximum I/O and fluctuation as the \sys{} under sequential workload.
\end{sloppypar}

The \clrs{} adversary dataset is another adversary dataset for the \bpt{}. In the dataset, we intentionally build a root-to-leaf path where all nodes are filled. This path raises the split propagation in the \bpt{} as well. Since the \clrs{} and the \bpt{} do not share the same tree structure as data keep inserting, their maximum I/O of the same tree height may not be the same. For example, at height three, the \bpt{} has peaked at 8 I/Os, which includes 2 node splits, while the \clrs{} reaches $3H+1=10$ that makes the tree height grow to four.

Under uniformly random dataset, the \bpt{} has the largest maximum I/O and fluctuation, the \clrs{} follows, and the \sys{} has the lowest one. This shows that the \sys{} performs consistently and keeps the fluctuation minimum across all baselines.

\subsection{Simulated Deletion}

We construct a full \bpt{} of height-4 where all nodes including the root node are 100\% full. The entire tree contains 4,096 entries. Then we issue 3,276 deletes that remove 80\% entries from the tree. 

Similar to the results of simulated insertion, the \sys{} shows the lowest fluctuation among all (Figures~\ref{fig:sim_delete_seq} and~\ref{fig:sim_delete_random}). Since the \bpt{} and the \clrs{} have the same deletion algorithm, their fluctuations are the same.

Under sequential dataset, the maximum I/O of the \bpt{} of height four is 12, including two node merges and one node borrowing. The largest fluctuation is 7 across all tree heights. the \sys{}'s maximum I/O is always $H+5$ and the fluctuation is four.

Under random dataset, the maximum I/O of the \bpt{} is higher. This is due to merges with left siblings that involve one write of the item deletion, one read of left sibling and one merge (write) with the left sibling. Therefore, the maximum I/O of the \bpt{} of height four is 14, i.e., $4H-2$. the \sys{}'s fluctuation is always capped at four.

We record the lowest node utilization and compute the fraction of underflowing nodes (utilization smaller than 50\%) after each delete in Figure~\ref{fig:sim_delete_util}. For \bpt{}, the lowest node utilization is always 50\% for both leaf and inner nodes; and there are no underflowing nodes. 
For \sys{} under random deletes in this experiment, the lowest observed leaf node utilization is 37.5\% and the lowest observed inner-node utilization is 25\%.
The fraction of underflowing node remains low for most of the deletes: only after around 73\% items are removed, half of the leaf nodes become underflowing; the inner nodes are more resistant to underflowing than leaf nodes, with fewer than 25\% becoming underflowing. For \sys{} under sequential deletes, the fraction of underflowing nodes is even smaller. Therefore, underflowing remains low until a substantial portion of items are removed and the delete algorithm converts the merge propagation into a slow-growing under-utilization.


\subsection{Mixed Inserts and Deletes}

We also conduct mixed inserts and deletes in Figure~\ref{fig:sim_mix}. All indexes are started empty and are issued with one million mixed operations. When the dataset contains 70\% inserts, the \sys{} shows a smaller fluctuation since tree height is four and higher, while the fluctuations of the \bpt{} and the \clrs{} keep increasing. Under the dataset of 90\% inserts, the fluctuations of the \bpt{} and the \clrs{} begin to surpass the fluctuation of the \sys{} at tree height 4 as well. The fluctuation of the \bpt{} grows quicker than the \clrs{}.

\section{Conclusion}


We introduce the \sys{}, a fluctuation-bounded \bpt{} that is equipped with a novel insert and delete algorithm. The insert algorithm eliminates split propagation and guarantees at most one node split per insertion. We formalize fluctuation for index operations, prove that the \sys{} bounds split cost to a constant, and demonstrate through simulation. The delete algorithm bounds the maximum delete cost by preventing the merge propagation.
The simulated results show that \sys{} has a bounded fluctuation under deletes as well as mixed inserts and deletes.



\bibliographystyle{acm}
\bibliography{sample}

\end{document}